\documentclass[11pt, a4paper]{article}
\usepackage{amsmath, amsthm, amssymb, url}
\usepackage[margin=1.1in]{geometry}
\usepackage[english]{babel}
\usepackage{multirow}
\usepackage{authblk}
\usepackage{tikz-cd} 
\usepackage{color}
\usepackage{hyperref}
\newcommand{\im}{\mathrm{im}}

\newcommand{\CA}{\mathrm{CA}}
\newcommand{\ICA}{\mathrm{ICA}}

\newcommand{\Fix}{\mathrm{Fix}}
\newcommand{\Aut}{\mathrm{Aut}}
\newcommand{\Z}{\mathbb{Z}}

\theoremstyle{plain}
\newtheorem{conjecture}{Conjecture}
\newtheorem{corollary}{Corollary}
\newtheorem{lemma}{Lemma}
\newtheorem{proposition}{Proposition}

\newtheorem{theorem}{Theorem}
\newtheorem{claim}{Claim}

\theoremstyle{definition}
\newtheorem{definition}{Definition}
\newtheorem{example}{Example}
\newtheorem{remark}{Remark}

\begin{document}

\title{One-dimensional cellular automata with a unique active transition}
\author[1]{Alonso Castillo-Ramirez\footnote{Email: alonso.castillor@academicos.udg.mx}}
\author[1]{Maria G. Maga\~na-Chavez \footnote{Email: maria.magana3917@alumnos.udg.mx }}
\author[1]{Luguis de los Santos Ba\~nos  \footnote{Email: luguis.banos@academicos.udg.mx }}
\affil[1]{Centro Universitario de Ciencias Exactas e Ingenier\'ias, Universidad de Guadalajara, M\'exico.}

\maketitle

\begin{abstract}
A one-dimensional cellular automaton $\tau : A^\mathbb{Z} \to A^\mathbb{Z}$ is a transformation of the full shift defined via a finite neighborhood $S \subset \mathbb{Z}$ and a local function $\mu : A^S \to A$. We study the family of cellular automata whose finite neighborhood $S$ is an interval containing $0$, and there exists a pattern $p \in A^S$ satisfying that $\mu(z) = z(0)$ if and only if $z \neq p$; this means that these cellular automata have a unique \emph{active transition}. Despite its simplicity, this family presents interesting and subtle problems, as the behavior of the cellular automaton completely depends on the structure of $p$. We show that every cellular automaton $\tau$ with a unique active transition $p \in A^S$ is either idempotent or strictly almost equicontinuous, and we completely characterize each one of these situations in terms of $p$. In essence, the idempotence of $\tau$ depends on the existence of a certain subpattern of $p$ with a translational symmetry.  
\\

\end{abstract}

\section{Introduction}\label{intro}

Cellular automata (CA) are models of computation over a discrete space with the key characteristic of being defined via local function that is applied homogeneously and in parallel over the whole space. Since their introduction by John von Neumann and Stanislaw Ulam, CA have gained significance due to their connections with discrete complex modeling, symbolic dynamics, and theoretical computer science.   

In particular, one-dimensional cellular automata are transformations $\tau : A^\mathbb{Z} \to A^\mathbb{Z}$, where $A$ is a set with at least two elements, known as the \emph{alphabet}, and $A^\Z$ is the set of all bi-infinite sequences of elements in $A$, known as the \emph{full shift}, such that there exists a \emph{finite neighborhood} $S \subset \mathbb{Z}$ and a \emph{local function} $\mu : A^S \to A$ satisfying  
\[ \tau(x)(k) = \mu( (k \cdot x)\vert_{S}), \quad \forall x \in A^\mathbb{Z}, k \in \mathbb{Z},  \]
where $k \cdot x \in A^\Z$ is the \emph{shift action} of $k$ on $x$, and $\vert_S$ denotes the restriction to $S$. The so-called Curtis-Hedlund-Lyndon theorem (see \cite{CSC10,LM95}) establishes that a transformation $\tau : A^\mathbb{Z} \to A^\mathbb{Z}$ is a cellular automaton if and only if it is uniformly continuous and shift equivariant. This result builds an important bridge between symbolic dynamics and the theory of cellular automata. 

In this paper, we consider the class of one-dimensional cellular automata $\tau : A^\mathbb{Z} \to A^\mathbb{Z}$ defined via a local function $\mu : A^S \to A$ whose finite neighborhood $S$ is an interval containing $0$, and there exists $p \in A^S$ such that $\mu(z) = z(0)$ if and only if $z \neq p$. This means that such a cellular automaton $\tau$ will act on $A^\Z$ as the identity function, except that when it reads a fixed pattern $p$ in a bi-infinite sequence, it will write a symbol $a:= \mu(p) \in A \setminus \{p(0)\}$. This is equivalent of saying that $\tau$ has a unique \emph{active transition} $p \in A^S$ in the terminology recently introduced by various authors in \cite{Pedro1, Pedro2, Concha, Fates1, Fates2}. 

The class of one-dimensional cellular automata with a unique active transition resemble the so-called \emph{Coven cellular automata}, whose remarkable dynamical properties were studied in \cite{Maas,Coven}. However, it follows by their construction that Coven cellular automata actually have two active transitions. 
 
In \cite{IDEM}, a wider class of cellular automata $\tau : A^G \to A^G$ (where the group of integers $\Z$ is replaced by an arbitrary group $G$) with a unique active transition $p \in A^S$ were studied. It was noticed that such CA are often \emph{idempotents}, which means that $\tau \circ \tau = \tau$; for example, it was shown that if $p$ is constant or \emph{symmetric} (in the sense that $S=S^{-1}$ and $p(s) = p(s^{-1})$, $\forall s \in S$), then $\tau$ is idempotent. Moreover, a full characterization of the idempotency of $\tau$ was given when $p$ is \emph{quasi-constant} (i.e. there is $s \in S$ such that $p \vert_{S \setminus \{s\}}$ is constant). However, the full characterization of the idempotency of $\tau : A^G \to A^G$ when $p$ is an arbitrary pattern remains open.   

The main result of this paper is the characterization of the idempotency of one-dimensional cellular automata with a unique active transition $p \in A^S$. Furthermore, we show that when the cellular automaton is not idempotent, then it follows a specific dynamic behavior. A cellular automaton $\tau : A^\Z \to A^\Z$ is a topological dynamical system, where the dynamics is given by the iteration of $\tau$. In this context, a cellular automaton is \emph{equicontinuous} if all bi-infinite sequences $x \in A^\Z$ are equicontinuous points in the sense that the trajectories of all neighbors of $x$ stay close to the trajectory of $x$. A cellular automaton is \emph{almost equicontinuous} if its equicontinuous points are dense in $A^\Z$, and it is \emph{strictly almost equicontinuous} if it is almost equicontinuous but not equicontinuous. 

Before we state the main result of this paper, we introduce some notation. For $k, \ell \in \mathbb{Z}$, we denote an interval by $[k, \ell] := \{ n \in \mathbb{Z} : k \leq n \leq \ell \}$. For a subset $S \subseteq \mathbb{Z}$, we use the notation $S_+ := \{ s \in S : s >0 \}$ and $S_- := \{ s \in S : s<0\}$.  

\begin{theorem}\label{intro-main}
Let $\tau : A^\mathbb{Z} \to A^\mathbb{Z}$ be a cellular automaton with a unique active transition $p \in A^S$, where $S := [k , \ell] \subset \mathbb{Z}$ is such that $k \leq 0 \leq \ell$. Let $\mu : A^S \to A$ be the corresponding local defining function for $\tau$. Then, $\tau$ is not idempotent if and only if there exists $t \in S_+$ such that 
\[ p(t) = \mu(p) \quad \text{ and } \quad p(i) = p(i-t), \  \forall i \in [k+t, \ell] \setminus \{ t \}, \]  
or there exists $t \in S_-$ such that
\[ p(t) = \mu(p) \quad \text{ and } \quad p(i) = p(i-t), \ \forall i \in [k, \ell + t] \setminus \{ t \}.  \]  
Moreover, $\tau$ is not idempotent if and only if it is strictly almost equicontinuous. 
\end{theorem} 

In essence, Theorem \ref{intro-main} characterizes the idempotence of $\tau$ in terms of the existence of $t \in S \setminus \{0\}$ such that $p(t) = \mu(p) \neq p(0)$ and there is a subpattern of $p$ invariant with respect to the translation by $t$. 

One-dimensional cellular automata satisfy a well-known dichotomy that they are either almost equicontinuous or \emph{sensitive}, which are considered to be chaotic systems \cite[Ch. 5]{Kurka}. In this spirit, Theorem \ref{intro-main} may be considered as a dichotomy satisfied by one-dimensional cellular automata with a unique active transition, as it shows that they are either idempotent, which are trivial as dynamical systems, or strictly almost equicontinuous, which are dynamical systems with intermediate complexity.    

Theorem \ref{intro-main} also has a relevant algebraic interpretation. The monoid $\CA(\Z,A) = \text{End}(A^\Z)$ consisting of the set all one-dimensional cellular automata equipped with the composition of functions has an intricate structure (see \cite{CSC10}), and its group of invertible elements (denoted by $\Aut(A^\Z)$ or $\ICA(\Z,A)$) has been widely studied in symbolic dynamics (e.g., see \cite[Sec. 13.2]{LM95}). In this context, Theorem \ref{intro-main} characterizes an infinite family of idempotents in $\CA(\Z,A)$, which is relevant as idempotents are fundamental tools in the theory of monoids (see \cite{semi}). Since it is known that $\tau \in \CA(\Z, A)$ is equicontinuous if and only if $\tau$ has \emph{finite order} (i.e. $\tau^m = \tau^n$, for some $m \neq n$), Theorem \ref{intro-main} also gives the algebraic dichotomy that every one-dimensional cellular automaton with a unique active transition is either idempotent or it has infinite order.     

The structure of this paper is as follows. In Section 2 we fix the notation and introduce some basic properties of one-dimensional cellular automata with a unique active transition. In Section 3 we present the proof of Theorem \ref{intro-main}. Finally, in Section 4, we discuss some open problems and future work.

\section{Cellular automata with a unique active transition}

Let $A$ be a set with at least two elements, known as the \emph{alphabet}. We shall assume that $\{0,1\} \subseteq A$. Although $A$ is usually assumed to be finite, the results of this paper also hold when it is infinite. The set $A^\mathbb{Z}$, known as the \emph{full shif}, is the set of all functions $x : \mathbb{Z} \to A$. We may identify each $x \in A^\mathbb{Z}$ with a bi-infinite sequence of symbols of $A$: 
\[  x = \dots x(-2) x(-1) x(0) x(1) x(2) \dots  \]
A \emph{pattern} over a finite subset $S \subset \mathbb{Z}$ is a function $p : S \to A$. We denote by $A^S$ the set of all patterns over $S$. We shall assume that finite subsets $S$ are given in the natural order of $\Z$; hence, if $S = \{ s_1, s_2, \dots s_n\}$, we assume $s_i < s_j$ for all $i < j$. With this convention, we may identify a pattern $p \in A^S$ with an element of the Cartesian power $A^{\vert S \vert}$ via
\[ p = p(s_1) p(s_2) \dots p(s_n). \] 

The \emph{shift action} of $\mathbb{Z}$ on $A^\mathbb{Z}$ is a function $\cdot : \mathbb{Z} \times A^\mathbb{Z} \to A^\mathbb{Z}$ defined by 
\[ (k \cdot x)(i) := x(i + k), \quad \forall x \in A^\mathbb{Z}, k, i \in \mathbb{Z}. \]
We say that a pattern $p \in A^S$ \emph{appears} in $x \in A^\Z$ if there exists $k \in \Z$ such that $(k \cdot x)\vert_S = p$.  
 
The full shift $A^{\mathbb{Z}}$ is a topological space with the product topology of the discrete topology of $A$. In general, $A^\mathbb{Z}$ is Hausdorff, totally disconnected and metrizable, and, when $A$ is finite, it is also compact (and so, a Cantor space). This topology is induced by the following metric $d$ on $A^\Z$ as remarked in \cite[1.9.2]{CSC10}: for all $x,y \in A^\Z$, 
\[ d(x,y) := \begin{cases}
0 & \text{ if } x=y \\
2^{- \max \{ k \geq 0 \; : \; x \vert_{[-k,k]} = y \vert_{[-k,k]} \} } & \text{ if } x \neq y.
\end{cases} \]

An \emph{endomorphism} of $A^\mathbb{Z}$ is a function $\tau : A^\mathbb{Z} \to A^{\mathbb{Z}}$  that is uniformly continuous (or just continuous, when $A$ is finite) and \emph{shift equivariant}, which means that it commutes with the shift action: 
\[ \tau( k \cdot x) = k \cdot \tau(x), \quad \forall x \in A^{\mathbb{Z}}, k \in \mathbb{Z}.  \]

By the Curtis-Hedlund-Lyndon theorem \cite[Sec. 1.8-1.9]{CSC10}, $\tau : A^\mathbb{Z} \to A^{\mathbb{Z}}$ is an endomorphism if and only if it is a \emph{cellular automaton}, which means that there is a finite subset $S \subset \mathbb{Z}$ (called a \emph{neighborhood}, or a \emph{memory set}) and a local function $\mu : A^S \to A$ such that 
\[ \tau(x)(k) = \mu( (k \cdot x)\vert_{S}), \quad \forall x \in A^\mathbb{Z}, k \in \mathbb{Z}.  \]

As defined in \cite[Def. 2.25]{Kurka}, an element $x \in A^\Z$ is an \emph{equicontinuous point} of a cellular automaton $\tau : A^\Z \to A^\Z$ if for all $\epsilon >0$ there exists $k \in \Z_+$ such that for all $y \in A^\Z$ with $y \vert_{[-k,k]} = x \vert_{[-k,k]}$ and all $n \geq 0$, we have $d(\tau^n(y) , \tau^n(x)) < \epsilon$. The cellular automaton $\tau$ itself is \emph{equicontinuous} if every $x \in A^\Z$ is an equicontinuous point of $\tau$. Moreover, we say that $\tau$ is \emph{almost equicontinuous} if its set of equicontinuous points is dense in $A^\Z$, and that $\tau$ is \emph{strictly almost equicontinuous} if it is almost equicontinuous but not equicontinuous.    

As recently introduced by several authors \cite{Pedro1, Pedro2, Concha, Fates1, Fates2}, the set of \emph{active transitions} of a local function $\mu : A^S \to A$ with $0 \in S$ is defined by
\[ \alpha(\mu) := \{z \in A^S : \mu(z) \neq z(0) \}. \]
We may think of this as the set of patterns in $A^S$ on which $\mu$ does not act as the projection to $0$. Clearly, the cellular automaton define by $\mu$ is equal to the identity function of $A^\Z$ if and only if $\alpha(\mu)= \emptyset$.  

\begin{definition}
Let $S \subset \Z$ be a finite set of integers such that $0 \in S$. We say that a cellular automaton $\tau : A^\mathbb{Z} \to A^{\mathbb{Z}}$ has a \emph{unique active transition} $p \in A^S$ if $\tau$ has a local defining function $\mu : A^S \to A$ such that $\mu(z) = z(0)$ if and only if $z \neq p$ (i.e. $\alpha(\mu) = \{ p \}$). 
\end{definition}

For a fixed $p \in A^S$, there exist $\vert A \vert - 1$ cellular automata with unique active transition $p$ since the local function $\mu : A^S \to A$ may be defined to satisfy $\mu(p) = a$ for any $a \in A \setminus \{p(0)\}$. 

\begin{example}
Let $A := \{ 0,1 \}$. One-dimensional cellular automata that admit a finite neighborhood $S := \{-1,0,1\} \subset \Z$ are called \emph{elementary cellular automata} (ECA) and they are labeled by a \emph{Wolfram number} (see \cite[Sec. 2.5]{Kari}). The local functions $\mu : A^S \to A$ of ECA may be described by a table whose first row contains the the patterns in $z \in A^S$ identified with triplets $z(-1)z(0)z(1) \in A^3$, and whose second row contains the the values $\mu(z) \in A$. For example, the local function of the ECA with Wolfram number 200 is defined by the following table:  
\[ \begin{tabular}{c|cccccccc}
$z \in A^S$ & $111$ & $110$ & $101$ & $100$ & $011$ & $010$ & $001$ & $000$ \\ \hline
$\mu(z) \in A$ & $1$ & $1$ & $0$ & $0$ & $1$ & $0$ & $0$ & $0$
\end{tabular}\] 
The Wolfram number 200 corresponds to the fact that the second row of the table is the binary number $11001000$, which corresponds to the decimal number $200$. Observe that this ECA has a unique active transition $p = 010 \in A^S$. 
\end{example}

We finish this section with a few results that hold for any cellular automaton with a unique active transition.

\begin{lemma}\label{le-b}
Let $\tau : A^\mathbb{Z} \to A^{\mathbb{Z}}$ be a cellular automaton with a unique active transition $p \in A^S$.
\begin{enumerate}
\item The set $S \subset \Z$ is the minimal neighborhood of $\tau$ (i.e. it is the neighborhood of smallest cardinality admitted by $\tau$). 
\item The set of fixed points $\Fix(\tau)$ is equal to the subshift $X_p$ of $A^\Z$ with forbidden pattern $p$: 
 \[ X_p := \{ x \in A^\Z : (k \cdot x) \vert_S \neq p, \ \forall k \in \Z \}. \]
\item $\tau : A^\Z \to A^\Z$ is almost equicontinuous. 
\end{enumerate}
\end{lemma}
\begin{proof}
The proof of parts (1.) and (2.) may be found in \cite{IDEM}. To show part (3.), we use \cite[Prop. 5.12]{Kurka} which establishes that $\tau$ is almost equicontinuous if and only if there exists an $r$-blocking word $u$ for $\tau$, where $r := \max\{ \max(S), \vert \min(S) \vert  \}$ is the \emph{radius} of the cellular automaton. Let $\mu : A^S \to A$ be the local defining function for $\tau$, and let $a:= \mu(p) \in A \setminus \{p(0)\}$. Define $u \in A^{[0,r]}$ by $u(i) := a$, for all $i \in [0,r]$. Then, for all $x \in A^\Z$ such that $x \vert_{[0,r]} = u$ and for all $i \in [0,r]$, we have 
\[ \tau(x)(i) = \mu( (i \cdot x) \vert_S) ) = (i \cdot x)(0) = x(i) = a,  \]
where the second equality follows since $(i \cdot x)\vert_{S} \neq p$ (as $(i \cdot x)(0) = a \neq p(0)$ for all $i \in [0,r]$). Repeating the above argument, we see that $\tau^n(x) \vert_{[0,r]} = u$ for all $n \geq 0$. This means that $u$ is an $r$-blocking word for $\tau$, and the result follows. 
\end{proof}

As we shall see in the next section, the subshift $X_p$ is important to determine the idempotence of cellular automata with a unique active transition $p \in A^S$. Such subshifts defined by forbidding a single pattern include the well-known \emph{golden mean subshift}, and have been recently studied in detail in \cite{Wu}. 


\section{Characterization of CA with a unique active transition} 

In general, the idempotence of a cellular automaton $\tau : A^\Z \to A^\Z$ with finite neighborhood $S \subset \Z$ may be checked by obtaining the local function of $\tau \circ \tau : A^\Z \to A^\Z$, which has finite neighborhood $S + S := \{ s+ t : s,t \in S \}$ (see \cite[Prop. 1.4.9]{CSC10}), and comparing it with the local function of $\tau$ (which may be extended to the domain $A^{S+S}$ in order to be compared). In the next results we shall see that there is a quicker and simpler method to decide the idempotence of a CA with a unique active transition.  

\begin{lemma}\label{cor-idem}
Let $\tau : A^\mathbb{Z} \to A^{\mathbb{Z}}$ be a cellular automaton with a unique active transition $p \in A^S$. Then, $\tau$ is not idempotent if and only if there exists $x \in A^\Z$ such that $p$ appears as a subpattern in $\tau(x) \in A^\Z$. 
\end{lemma}
\begin{proof}
In general, it holds that a transformation $\tau : A^\Z \to A^\Z$ is idempotent if and only if $\Fix(\tau) = \im(\tau)$, so by Lemma \ref{le-b} (2), this is equivalent to $X_p = \im(\tau)$. Therefore, $\tau$ is not idempotent if and only if $X_p$ is a proper subset of $\im(\tau)$, which is equivalent to the statement of this lemma. 
\end{proof}

\begin{lemma}\label{le-2}
Let $\tau : A^\mathbb{Z} \to A^{\mathbb{Z}}$ be a cellular automaton with a unique active transition $p \in A^S$. Let $\mu : A^S \to A$ be the corresponding local defining function of $\tau$. Then, $\tau$ is not idempotent if and only if there exists $x \in A^\Z$ such that 
\begin{equation} \label{eq-x}
\mu( (s \cdot x) \vert_S  ) = p(s), \quad \forall s \in S. 
\end{equation}
Furthermore, if such $x \in A^\Z$ exists, it must also satisfy the following:
\begin{enumerate}
\item $x \vert_S \neq p$,
\item $x(0) = p(0)$. 
\item $(t \cdot x) \vert_S  = p$ for some $t \in S \setminus \{0\}$. 
\end{enumerate}
\end{lemma}
\begin{proof}
By Corollary \ref{cor-idem}, $\tau$ is not idempotent if and only if there exists $x \in A^\Z$ such that $p$ appears as a subpattern in $\tau(x)$. Since $\im(\tau)$ is invariant under the shift action, we may assume that $p$ appears in $S$, i.e.
\[ \tau(x)(s) = p(s), \quad \forall s \in S. \]
Hence, in terms of the local defining function
\[ \mu((s \cdot x) \vert_S) = \tau(x)(s) = p(s), \quad \forall s \in S.  \]
We must have $x \vert_S \neq p$, since $\mu(p) \neq p(0)$, and so $p(0) = \mu(x \vert_S) = x(0)$. Finally, suppose that $(s \cdot x) \vert_S \neq p$ for all $s \in S$. Then, $p(s) = \mu((s \cdot x) \vert_S) = x(s)$ for all $s \in S$, contradicting that $x \vert_S  \neq p$. Thus, there must exists $t \in S \setminus \{0\}$ with $(t \cdot x) \vert_S  = p$. 
\end{proof}

\begin{remark}
The subpatterns $(s \cdot x) \vert_S$ for $s \in S$ are completely determined by $x \vert_{S+S}$. Hence, Lemma \ref{le-2} also gives a quick algorithm for deciding when a cellular automaton generated by a pattern is idempotent, which consists on checking if there exists a pattern $q \in A^{S+S}$ such that $\mu(q_s ) = p(s)$ for all $s \in S$, where $q_s \in A^S$ is the pattern defined by $q_s(t) := q(s+t)$, $\forall t \in S$.   
\end{remark}

\begin{example}
For $S = \{-1,0,1\}$, a cellular automaton with unique active transition $p = 001 \in A^S$ is not idempotent because there exists $q = 00001 \in A^{S+S}$ such that 
\[ \mu(q_{-1}) \mu(q_0) \mu(q_1) = \mu(000) \mu(000) \mu(001) = 001 = p.  \]
\end{example}

\begin{example}\label{counter}
Unlike some constructions and results in symbolic dynamics, the idempotence of a CA with a unique active transition depends on both the pattern $p \in A^S$ and the local neighborhood $S \subset \Z$. For example, the CA with unique active transition $p_1=0010 \in A^{S_1}$ with $S_1 := [-1,2]$ is idempotent, while the CA with unique active transition $p_2 = 0010 \in A^{S_2}$ with $S_2 := [0,3]$ is not idempotent.   
\end{example}

\begin{corollary}\label{cor-1}
With the notation of the previous lemma,
 \[ p(t)= \mu(p) \neq p(e). \]
\end{corollary}
\begin{proof}
Since $(t \cdot x) \vert_S = p$ for $t \in S \setminus \{0\}$, then 
\[ p(t) = \mu( (t \cdot x) \vert_S ) = \mu(p).   \]
It is clear that $\mu(p) \neq p(0)$, because $p$ is an active transition of $\mu$. 
\end{proof}

The previous results easily generalize to cellular automata $\tau : A^G \to A^G$, where $G$ is any abstract group. However, henceforth, our proofs will rely on specific properties of the group of integers $\Z$.  

Our next goal is to show that there is a \emph{unique} $t \in S \setminus \{0\}$ such that $(t \cdot x) \vert_S  = p$, where $x \in A^\Z$ satisfies the condition (\ref{eq-x}) of Lemma \ref{le-2}. The uniqueness of $t$ shall be crucial in the proof of Theorem \ref{intro-main}, but first we show a couple of technical results that hold in the case when $S$ is an interval of $\Z$ containing $0$. 

\begin{lemma}\label{rk-int}
Let $S \subset \Z$ be an interval with $0 \in S$, and let $t,r \in S$, $t < r$. For each $i \in (t,r)$, either $i-t \in S$ or $i-r \in S$. 
\end{lemma}
\begin{proof}
Suppose that $i-r \not \in S$. Since $ i-r <0$, we must have that $i - r < t$. But now, $0 < i-t < r $, implies that $i - r \in S$, since $0, r \in S$ and $S$ is an interval. 
\end{proof}

\begin{lemma}\label{le-seq}
Let $S \subset \mathbb{Z}$ be an interval such that $0 \in S$. Let $x \in A^\Z$ and $p \in A^S$. Suppose that there exist $t,r \in S$, $t < r$, satisfying that $(t \cdot x) \vert_S  = (r \cdot x) \vert_S  = p$, and
\begin{equation}\label{hyp-seq}
\forall s \in (t,r), \ \epsilon \in \{r,t\},  \quad   ( s - \epsilon \in S)  \ \Rightarrow  \ p(s) = p(s-\epsilon).  
\end{equation}
Then, $p(0) = p(t)$. 
\end{lemma}
\begin{proof}
Define $m := r-t \in \mathbb{Z}_+$. We first show that since $(t \cdot x) \vert_S  = (r \cdot x) \vert_S  = p$, then $p : S \to A$ must be constant on the classes modulo $m$ that intersect $S$. Denote by $[s]_m$ the equivalence class modulo $m$ of $s \in \mathbb{Z}$. 

\begin{claim}\label{lemma-mod}
If $s_1, s_2 \in S$ and $[s_1]_m = [s_2]_m$, then $p(s_1) = p(s_2)$.
\end{claim}
\begin{proof}
Observe that for all $s \in S$ such that $s+m \in S$, we have 
\[ p(s) = (r \cdot x)(s) = x(s+ r) = (t \cdot x)(s +r - t) = p(s + m). \]
If $s \in S$ and $s + km \in S$, for $k \in \mathbb{Z}_+$, then $s+m, s+2m, \dots, s+(k-1)m \in S$, as $S$ is an interval. It follows that $p(s) = p(s+km)$. 
On the other hand, if $s \in S$ and $s-m \in S$, then
\[ p(s) = (t \cdot x)(s) = x(s + t) = (r\cdot x) (s - m) = p(s-m). \]
It follows that if $s \in S$ and $s+km \in S$, for $k \in \mathbb{Z}_-$, then $p(s) = p(s+km)$. 
\end{proof}

If $[t]_m = [0]_m$, then Claim \ref{lemma-mod} implies that $p(t) = p(0)$, and the result follows. Hence, assume that $[ t]_m \neq [0]_m$. 

Let $n$ be the group-theoretic order of $[t]_m$ in $\mathbb{Z}_m$; this is, $n$ is the smallest positive integer such that $n[t]_m = 0$. Let $I := [t,r) \subset S$. Since $\vert I \vert = m$ and $I$ consists of consecutive integers, then $I$ contains representatives of all classes modulo $m$.  Let $s_1, s_2, s_3, \dots, s_{n} \in I$ be such that 
\[ [s_j]_m = j [t]_m , \quad \forall j \in [1,n].  \]
By Lemma \ref{rk-int}, for each $j \in [2,n]$, there exists $\epsilon_j$ such that $s_j - \epsilon_j \in S$, with $\epsilon_j \in \{t,r\}$, and by hypothesis (\ref{hyp-seq})
\[  p(s_j) = p(s_j - \epsilon_j), \quad \forall j \in [2,n].  \]  
Observe that, for all $j \in [2,n]$, we have $[\epsilon_j]_m = [t]_m$, so
\[ [s_j - \epsilon_j]_m = [s_j]_m - [\epsilon_j]_m = j[t]_m - [t]_m = (j-1)[t]_m = [s_{j-1}]_m.  \]
Therefore,
\[  p(s_j - \epsilon_j) = p(s_{j-1}), \quad \forall j \in [2,n].  \]
This shows that $p(s_n) = p(s_1)$. However, $s_1 = t$ and $[s_n] = n[t] = [0]$, so by Claim \ref{lemma-mod},
\[ p(0) = p(s_n) = p(s_1) = p(t).     \] 
The result follows. 
\end{proof}

\begin{theorem}\label{uniqueness}
Let $S \subset \mathbb{Z}$ be an interval such that $0 \in S$. Let $\tau : A^\mathbb{Z} \to A^{\mathbb{Z}}$ be a cellular automaton with a unique active transition $p \in A^S$ and with corresponding local defining function $\mu : A^S \to A$. Then, for every $x \in A^\Z$ satisfying $\mu((s \cdot x)\vert_S) = p(s)$, $\forall s \in S$, there exists a unique $t \in S \setminus \{0\}$ such that $(t \cdot x)\vert_S = p$.
\end{theorem}
\begin{proof}
The existence of such $t \in S \setminus \{0\}$ is given by Lemma \ref{le-2}, part (3). In order to prove the uniqueness, suppose there are $t, r \in S \setminus \{0\}$ such that $(t \cdot x)\vert_S = (r \cdot x)\vert_S = p$. Without loss, assume that $t< r$. We may also assume that 
\begin{equation}\label{eqhyp}
 (s \cdot x)\vert_S \neq p, \quad \forall s \in (t,r) \subset S.    
\end{equation}
This implies that for each $s \in (t,r)$ and $\epsilon \in \{t,r\}$ such that $s - \epsilon \in S$, then 
\[ p(s) = \mu((s \cdot x)\vert_S) = (s \cdot x)(0) = x(s) = (\epsilon \cdot x)( s -  \epsilon) = p(s - \epsilon), \] 
where the last equality holds since $ (\epsilon \cdot x) \vert_S =p$. 

The above shows that the hypothesis (\ref{hyp-seq}) of Lemma \ref{le-seq} is satisfied, so $p(0) = p(t)$. However, by Corollary \ref{cor-1}, $p(t) = \mu(p)$, which contradicts that $\mu(p) \neq p(0)$. The result follows. 
\end{proof}

\begin{example}\label{counter}
Theorem \ref{uniqueness} fails when $S \subset \mathbb{Z}$ is not an interval. For example, consider $S:= \{-1, 0 , 3\}$ and the pattern $p \in A^S$ defined by
\[ p =\left( \begin{tabular}{ccc}
$-1$ & $0$ & $3$ \\
$0$ & $1$ & $0$ 
\end{tabular} \right).   \]
Then, the cellular automaton $\tau : A^\Z \to A^\Z$ with a unique active transition $p$ is not idempotent; an $x \in A^\Z$ that satisfies the condition of Lemma \ref{le-2} is given by
\[ x = \dots 011010 \dots \]
However, we see that there exist two $s \in S$ such that $(s \cdot x)\vert_S = p$; namely, $s=-1$ and $s=3$.  
\end{example}

 \begin{theorem}\label{th-main}
 Let $S := [k, \ell] \subset \mathbb{Z}$ be such that $k \leq 0 \leq l$. Let $\tau : A^{\mathbb{Z}} \to A^{\mathbb{Z}}$ be the cellular automaton with a unique active transition $p \in A^S$, with corresponding local function $\mu : A^S \to A$. Then, $\tau$ is not idempotent if and only if there exists $t \in S_+$ such that 
\[ p(t) = \mu(p) \quad \text{ and } \quad p(i) = p(i-t), \  \forall i \in [k+t, \ell] \setminus \{ t \}, \]  
or there exists $t \in S_-$ such that
\[ p(t) = \mu(p) \quad \text{ and } \quad p(i) = p(i-t), \ \forall i \in [k, \ell + t] \setminus \{ t \}.  \]  
 \end{theorem}
 \begin{proof}
 Let $a := \mu(p) \neq p(0)$. Suppose that $\tau$ is not idempotent. By Lemma \ref{le-2}, there exists $x \in A^\Z$ such that $\mu((s\cdot x)\vert_S) = p(s)$ for all $s \in S$, and $t \in S \setminus \{0\}$ such that $(t \cdot x) \vert_S = p$. By Theorem \ref{uniqueness}, this $t \in S \setminus \{0\}$ is unique. Suppose that $t \in S_+$. By Corollary \ref{cor-1}, $p(t) = a$. Fix $i \in [k+t,t)$. By uniqueness of $t$, we have $(i \cdot x)\vert_S \neq p$, and since $i-t \in [k,0) \subset S$, it follows that
\[ p(i) = \mu((i \cdot x)\vert_S) = (i \cdot x)(0) = x(i) = (t \cdot x) (i-t) = p(i-t), \quad \forall i \in [k + t, t). \]
Now fix $j \in (t, \ell]$. By uniqueness of $t$, we have $(j \cdot x)\vert_S \neq p$, and since $j-t \in (0, \ell-t] \subset S$, it follows that  
\[ p(j) = \mu((j \cdot x)\vert_S) = (j \cdot x) (0) = x(j) = (t \cdot x)(j -t ) = p(j - t), \quad \forall j \in (t, \ell].  \]
The result follows analogously when $t \in S_-$, or may be shown by applying the well-known reflection automorphism of the monoid $\CA(\Z,A)$ (e.g, see \cite[Sec. 3]{CRG}).  

Suppose now that $p \in A^S$ satisfies that there exist $t \in S_+$ such that 
\begin{equation}\label{hyp}
 p(t) = a \quad \text{ and } \quad p(i) = p(i-t), \  \forall i \in [k+t, \ell] \setminus \{ t \}. 
\end{equation}

Define $x \in A^\Z$ as follows: for all $j \in \mathbb{Z}$,
\[ x(j) := \begin{cases}
p(j + nt) & \text{if } \exists n \geq 0 \text{ s. t. } j - nt \in [k , k + t) \\
p(j - t) & \text{if } j \in [k + t, \ell + t] \\
p(j - nt) & \text{if } \exists n \geq 2 \text{ s. t. } j-nt \in (\ell-t, \ell].
\end{cases} \]
As a bi-infinite sequence, we see that 
\[ x = (p\vert_{[k,k+t)})^\infty \ p \ (p \vert_{(\ell-t, \ell]})^\infty, \text{ where } x(t) = p(0),   \]
where the notation $q^\infty$ means that a pattern $q$ is repeated periodically infinitely many times. 

\begin{claim} \label{claim-x}
Observe that $x$ satisfies 
\[ x(j) = x( j -t), \ \forall j > t \]
\end{claim}
\begin{proof}
This follows by construction of $x$, and since,
\[  x(j) = p(j - t) = p(j) = x(j+t), \quad \forall j \in (t , \ell - t], \]
\end{proof}

 \begin{claim}\label{uniqueness2}
 With the above construction, $(t \cdot x) \vert_S = p$ and $(r \cdot x) \vert_S \neq p$, $\forall r \in S \setminus \{t\}$. 
 \end{claim}
\begin{proof}
It is clear by definition that $(t \cdot x) \vert_S = p$. Suppose that $(r \cdot x) \vert_S = p$ for $r \neq t$. Then
\[ p(0)= (r \cdot x)(0)=x(r) \quad \text{and} \quad p(t) =(r \cdot x)(t) = x(r+t).   \]
If $r \in (t , \ell]$, by Claim \ref{claim-x}, $x(r+t) = x(r)$, so $p(t) = p(0)$, which is a contradiction. 

Suppose that $r \in [k, t)$. We will show that for $i \in (r,t)$ and $\epsilon \in \{ r,t\}$ such that $i- \epsilon \in S=[k, \ell]$, then $p(i) = p(i- \epsilon)$. If $i-t \in S$, then $p(i) = p(i-t)$ by hypothesis (\ref{hyp}). Suppose that $i - r \in [k, \ell]$ and $i - t \not\in [k,\ell]$. The latter condition implies $i - t < k$, so $i \in (r, k + t) \subset [k, k+t)$. Then, by the construction of $x$: 
\[ p(i) = x(i) = (r \cdot x)(i-r) = p(i-r). \]
The above shows that the hypothesis (\ref{hyp-seq}) of Lemma \ref{le-seq} is satisfied, so $p(0) = p(t)$, which is a contradiction. The result of the claim follows.  
\end{proof}

Now we prove that $\mu(q_s) = p(s)$ for all $s \in S$. For $s = t$, $q_t = p$ implies that
\[ p(t) =  a  = \mu(p) =  \mu((t \cdot x) \vert_S).  \]
For all $s \in S \setminus \{t\}$, Claim \ref{uniqueness2} shows that $(s \cdot x) \vert_S \neq p$, so
\[ \mu((s \cdot x) \vert_S) = (s \cdot x)(0) = x(s). \]
So the result would follow if we show that $x(s) = p(s)$ for all $s \in S \setminus \{t\}$. If $s \in [k,k+t)$, it is clear from the construction of $x$ that $x(s) = p(s)$. If $s \in [k+t,\ell] \setminus \{ t \}$, the construction of $x$ and hypothesis (\ref{hyp}) imply
\[ x(s) = p(s-t) = p(s).   \] 
The result follows. 
\end{proof}

\begin{example}
For $S=[-2,2]$, Theorem \ref{th-main} shows that a one-dimensional cellular automaton with a unique active transition $p=01010$ is idempotent, while one with a unique active transition $p=00101$ is not idempotent (in this case $t=-2$). 
\end{example}

\begin{example}\label{ex-not}
Theorem \ref{th-main} fails when $S$ is not an interval. For example, the cellular automaton with a unique active transition $p \in A^S$ given in Example \ref{counter} is not idempotent but does not satisfy the conditions of Theorem \ref{th-main}.  
\end{example}

\begin{proposition}\label{cor-x}
Let $S \subset \mathbb{Z}$ be an interval such that $0 \in S$. Let $\tau : A^{\mathbb{Z}} \to A^{\mathbb{Z}}$ be the cellular automaton with a unique active transition $p \in A^S$. If $\tau$ is not idempotent, then there exists $x \in A^{\mathbb{Z}}$ such that 
\[ \tau^n(x) \neq \tau^m(x), \quad \forall n,m \in \mathbb{Z}_+, \ n \neq m. \] 
\end{proposition}
\begin{proof}
 Let $S = [k, \ell] \subset \mathbb{Z}$ be such that $k \leq 0 \leq \ell$. Let $\mu : A^S \to A$ be the local rule generated by the pattern $p$. Let $a := \mu(p) \in A \setminus \{ p(e) \}$.
By Theorem \ref{th-main},  there exist $t \in S_+$ such that 
\[  p(t) = a \quad \text{ and } \quad p(i) = p(i-t), \  \forall i \in [k+t, \ell] \setminus \{ t \}. \]
This implies that $p \vert_{[k,0)}=p \vert_{[k+t,t)}$ and $p \vert_{(0, \ell - t]} = p \vert_{(t, \ell]}$.

We consider $x \in A^{\mathbb{Z}}$ as defined in the proof of Theorem \ref{th-main}:
 \[ x = (p\vert_{[k,k+t)})^\infty \ p \ (p \vert_{(\ell-t, \ell]})^\infty, \text{ where } x(t) = p(0). \]
We will show that $p$ only appears once as a subpattern in $x$. 
 \begin{claim}\label{uniqueness3}
  With the above construction, $(r \cdot x) \vert_S \neq p$, $\forall r \in \Z \setminus \{t\}$. 
 \end{claim}
 \begin{proof}
By Claim \ref{uniqueness2}, $(r \cdot x) \vert_S \neq p$, $\forall r \in [k, \ell] \setminus \{t\}$, so suppose that there exists $r > \ell$ or $r < k$ such that $(r \cdot x) \vert_S = p$. Then,
\[  (r \cdot x)(0)= x(r) = p(0) \quad \text{and} \quad (r \cdot x)(t)= x(r + t) = p(t).  \]
Recall that the construction of $x$ satisfies that $x(i) = x(i - nt)$ for all $i < k+t$, $n \in \mathbb{Z}_+$ and $x(i) = x(i+nt )$ for all $i > \ell - t$, $n \in \mathbb{Z}_+$. 

If $r < k$, then $r + t < k + t$, so 
\[ p(t) = x(r + t) = x(r + t - t) = x(r) =p(0),  \]
which is a contradiction. If $r > \ell$, then $r > \ell - t$, so
\[ p(0) = x(r) = x(r + t) = p(t), \]
which is also a contradiction. The result follows. 
 \end{proof}
 
The previous claim implies that all the terms in $\tau(x) \in A^\Z$ are the same as in $x \in A^\Z$ except that $x(t) = p(0)$ and $\tau(x)(t) = a$. Therefore,
\begin{align*}
 \tau(x) & = (p\vert_{[k,k+t)})^\infty \ p\vert_{[k,0)} \ a \ p \vert_{(0,\ell]} \ (p \vert_{(\ell-t, \ell]})^\infty  \\
 & = (p\vert_{[k,k+t)})^\infty \ p \vert_{[k+t,t)} \ p(t) \ p\vert_{(0, \ell - t]} \ p \vert_{(\ell-t, \ell]} (p \vert_{(\ell-t, \ell]})^\infty \\
 & = (p\vert_{[k,k+t)})^\infty \ p\vert_{[k,k+t)}\ p \vert_{[k+t,t)} \ p(t) \ p \vert_{(t, \ell]} \ (p \vert_{(\ell-t, \ell]})^\infty \\
& = (p\vert_{[k,k+t)})^\infty \ p \ (p \vert_{(\ell-t, \ell]})^\infty
\end{align*}
where $\tau(x)(0) = p(0)$. This shows that
\[ \tau(x) =  (-t) \cdot x. \]
Iterating, we see that
\[ \tau^n(x) = (-nt) \cdot x, \quad \forall n \in \mathbb{Z}_+.  \]
As $x$ is aperiodic (since $p$ appears only once as a subpattern), the result follows. 
\end{proof}

\begin{remark}
The proof of Theorem \ref{th-main} follows easily if we define
\[ x =  (a_k)^\infty \ (p\vert_{[k,k+t)}) \ p \ (a_\ell)^\infty, \text{ where } x(t) = p(0),   \]
and $a_k, a_\ell \in A$ are fixed symbols such that $a_k \neq p(k)$ and $a_\ell \neq p(\ell)$. However, the $x \in A^\Z$ that was actually chosen in the proof allows us to directly deduce Proposition \ref{cor-x}. 
\end{remark}

The \emph{space-time diagram} of a cellular automaton $\tau : A^\Z \to A^\Z$ with initial configuration $x \in A^\Z$ is a function $z : \mathbb{Z} \times \mathbb{N} \to A$ defined by $z(i,n) := \tau^n(x)(i)$, for all $n \in \mathbb{N}$ and $i \in \mathbb{Z}$. In the function $z$, we consider the input $n$ as the time and we allow $n=0$. When $A=\{0,1\}$, we shall depict a space-time diagram as a two-dimensional regular grid whose cells $(i,n)$ are colored in black when $z(i,n)=1$ and in white when $z(i,n) = 0$.  

\begin{example}
For $S=[-3,3]$ and $A=\{0,1\}$, the cellular automaton $\tau : A^\Z \to A^\Z$ with unique active transition $p=1101100 \in A^S$ is not idempotent as it satisfies condition of Theorem \ref{th-main} with $t = 3$. Figure \ref{fig1} shows the space-time diagram of $\tau$ with initial configuration $x \in A^\Z$ as given in the proof of Proposition \ref{cor-x}.
\begin{figure}
\label{fig1}
\centering
\includegraphics[scale=.9]{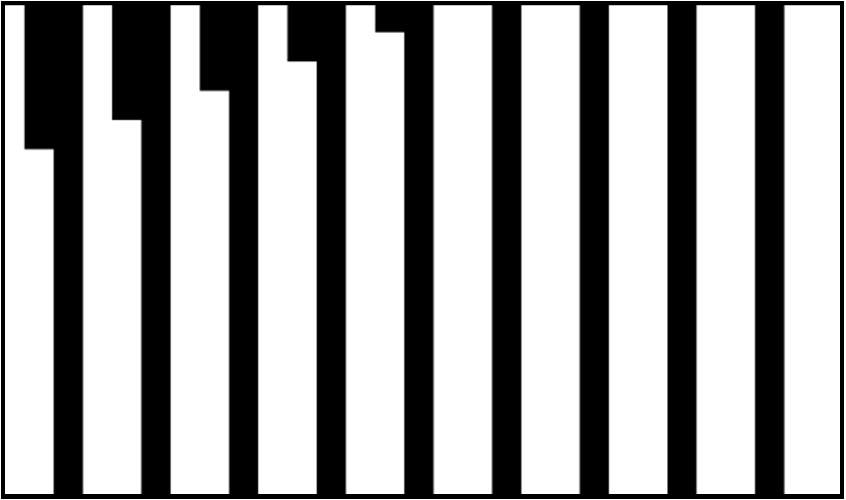}
\caption{Space-time diagram of CA with unique active transition $p=1101100$ and $S=[-3,3]$.}
\end{figure}
\end{example}

\begin{corollary}
Let $S \subset \mathbb{Z}$ be an interval such that $0 \in S$. Let $\tau : A^{\mathbb{Z}} \to A^{\mathbb{Z}}$ be the cellular automaton with a unique active transition $p \in A^S$. Then, $\tau$ is not idempotent if and only if $\tau$ is strictly almost equicontinuous. 
\end{corollary}
\begin{proof}
This follows by Proposition \ref{cor-x}, Lemma \ref{le-b}, and the fact that $\tau$ is equicontinuous if and only if $\tau^m = \tau^n$, for some $m \neq n$ (see \cite[Ch. 5]{Kurka}). 
\end{proof}

\begin{example}\label{counter2}
For $S=\{-1,0,3\}$ and $A=\{0,1\}$, the cellular automaton $\tau : A^\Z \to A^\Z$ with unique active transition $p=010 \in A^S$ is not idempotent, as noted in Example \ref{counter}. Figure \ref{fig2} shows a space-time diagram of $\tau$ with initial configuration
\[ x = (011)^\infty 01 0  0^\infty , \]
in which the pattern $p$ appears at every moment of time, which implies that $\tau$ is strictly almost equicontinuous. However, this may not be deduced from Proposition \ref{cor-x} since $S$ is not an interval.
\begin{figure}
\label{fig2}
\centering
\includegraphics[scale=.88]{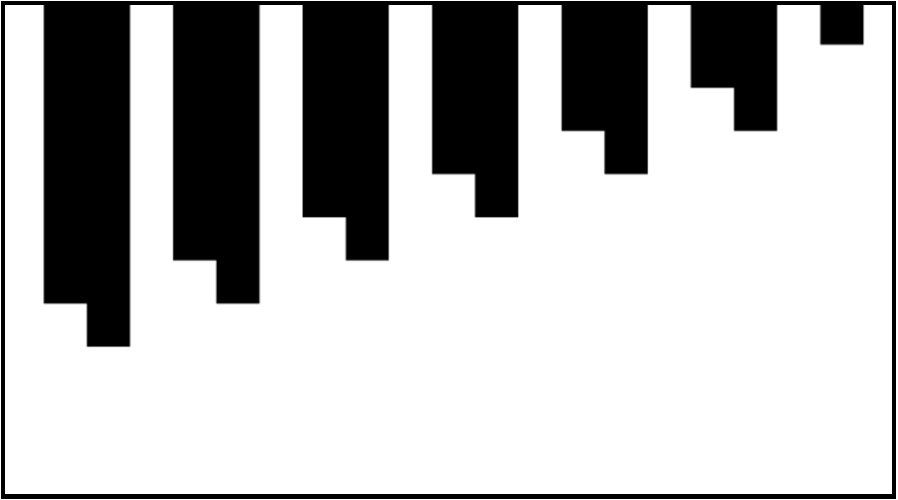}
\caption{Space-time diagram of CA with unique active transition $p=010$ and $S=\{-1,0,3\}$.}
\end{figure}
\end{example}


\section{Conclusions}

We characterized the idempotency of a one-dimensional cellular automaton $\tau$ with a unique active transition $p$ in terms of the existence of a subpattern of $p$ with a translational symmetry, and we showed that such $\tau$ is not idempotent if and only if it is strictly almost equicontinuous. Therefore, when $\tau$ is idempotent, the pattern $p$ will never appear after the first step in the space-time diagram of $\tau$, but when $\tau$ is not idempotent, for certain configurations of $A^\Z$ the pattern $p$ will appear at all times in the space-time diagram of $\tau$. 

We believe that the most important lines for future work is to generalize Theorem \ref{intro-main} to one-dimensional CA with unique active transition $p \in A^S$, where the finite neighborhood $S \subset \Z$ is not an interval, and to multidimensional CA. We we propose the following conjecture which generalizes Theorem \ref{intro-main} to both of the previous scenarios. 

\begin{conjecture}
For $d \in \mathbb{Z}_+$, let $\tau : A^{\Z^d} \to A^{\Z^d}$ be a cellular automaton with a unique active transition $p \in A^S$, where $S \subset \Z^d$ is any finite subset such that $0 \in S$. Then, $\tau$ is not idempotent if and only if there exists $t \in S$ such that 
\[  p(t) \neq p(e) \text{ and } p \vert_{U \setminus \left\{t \right\} } = p_{(-t + U) \setminus \{ 0 \}}, \]
where $U:= S \cap (t + S)$. 
\end{conjecture}

Interestingly, the previous conjecture may be written using the \emph{self-agreement set} of a pattern $p \in A^S$ which was independently introduced in \cite[Definition 5.1]{Wu} in the study of the subsifts $X_p$ with a unique forbidden pattern. However, we stress that the proof of this conjecture must involve more than the study of the subshift $X_p$, since, as show in Example \ref{counter}, there exist patterns $p_1 \in A^{S_1}$ and $p_2 \in A^{S_2}$ such that $X_{p_1} = X_{p_2}$, but the cellular automaton with unique active transition $p_1$ is idempotent, withe the one with unique active transition $p_2$ is not idempotent.  

It was noted in private commutation by Edgar Alcal\'a Arroyo that the cellular automaton $\tau : A^\Z \to A^\Z$ with unique active transition $p=010 \in A^S$, with $S=\{0,1,3\}$, is neither idempotent nor has infinite order:  it satisfies that $\tau^2 \neq \tau$ and $\tau^3 = \tau^2$. This is an example of an equicontinuous one-dimensional cellular automaton with a unique active transition whose finite neighborhood is not an interval and it is not idempotent. Therefore, another interesting direction for future work, is to characterize the equicontinuity of such classes of one-dimensional cellular automata.   

We also believe that it is interesting to study the implications or interpretations that Theorem \ref{intro-main} may have in the modeling of discrete complex systems. For example, the active transition $p$ may represent a certain rule of a genetic mutation (e.g., see \cite{Chen}). In this context, active transitions associated with non-idempotent CA may be interpreted as a more persistent mutation rule due to the fact that the cellular automaton is strictly almost eqicontinuous.



\section*{Acknowledgments}

The second author was supported by CONAHCYT \emph{Becas nacionales para estudios de posgrado}. The third author was supported by a CONAHCYT Postdoctoral Fellowship \emph{Estancias Posdoctorales por M\'exico}, No. I1200/320/2022.



\begin{thebibliography}{}

\bibitem{Pedro1} Balbi, P. P., Mattos, T., Ruivo, E.: Characterisation of the elementary cellular automata with
neighbourhood priority based deterministic updates. Comun. Nonlinear Sci. Numer. Simulat. \textbf{104} (2022) 106018. \url{https://doi.org/10.1016/j.cnsns.2021.106018}

\bibitem{Pedro2} Balbi, P. P., Mattos, T., Ruivo, E.: From Multiple to Single Updates Per Cell in Elementary Cellular Automata with Neighbourhood Based Priority. In: Das, S., Roy, S., Bhattacharjee, K. (eds) The Mathematical Artist. Emergence, Complexity and Computation \textbf{45}, Springer, Cham., 2022. \url{https://doi.org/10.1007/978-3-031-03986-7_6}

\bibitem{Maas} Blanchard, F., Maass, A.: Dynamical behaviour of Coven's aperiodic cellular automata. Theoretical Computer Science \textbf{163.1-2} (1996) 291-302.

\bibitem{CRG} Castillo-Ramirez, A., Gadouleau, M.: Elementary, Finite and Linear vN-Regular Cellular Automata, Inf. Comput. \textbf{274} (2020) 104533 \url{https://doi.org/10.1016/j.ic.2020.104533}

\bibitem{IDEM} Castillo-Ramirez, A., Magaña-Chavez, M. G, Veliz-Quintero, E.: Idempotent cellular automata and their natural order, Theoret. Comput. Sci. \textbf{1009} (2024) 114698. \url{https://doi.org/10.1016/j.tcs.2024.114698} 

\bibitem{CSC10} Ceccherini-Silberstein, T., Coornaert, M.: Cellular Automata and Groups. Springer Monographs in Mathematics, Springer-Verlag Berlin Heidelberg, 2010.

\bibitem{Wu} Chandgotia, N., Marcus, B., Richey, J., Wu, C.: Shifts of Finite Type Obtained by Forbidding a Single Pattern (2024). \url{https://doi.org/10.48550/arXiv.2409.09024}.

\bibitem{Chen} Chen, L., et. al.: Recognizing Pattern and Rule of Mutation Signatures Corresponding to Cancer Types. Front. Cell Dev. Biol., 25 August 2021. Sec. Molecular and Cellular Pathology \textbf{9} (2021). https://doi.org/10.3389/fcell.2021.712931

\bibitem{semi} Clifford, A.H., Preston, G.B.: The algebraic theory of semigroups. 2nd Ed. Mathematical Surveys and Monographs \textbf{7.I} American Mathematical Society, Providence, 1964.

\bibitem{Concha} Concha-Vega, P., Goles, E., Montealegre, P., Ríos-Wilson, M., \& Santivañez, J.: Introducing the activity parameter for elementary cellular automata. Int. J. Mod. Phys. C \textbf{33(9)} (2022) 2250121. \url{https://doi.org/10.1142/S0129183122501212}

\bibitem{Coven} Coven, E. M.: Topological entropy of block maps. Proceedings of the American Mathematical Society \textbf{78.4} (1980) 590-594.

\bibitem{Fates1} Fates, N.: A tutorial on elementary cellular automata with fully asynchronous updating. Nat. Comput. \textbf{19} (1) (2020) 179--197. \url{https://doi.org/10.1007/s11047-020-09782-7}

\bibitem{Fates2} Fates, N.: Asynchronous Cellular Automata. In: Meyers, R. (eds) Encyclopedia of Complexity and Systems Science. Springer, Berlin, Heidelberg, 2018. \url{https://doi.org/10.1007/978-3-642-27737-5_671-2}

\bibitem{Kari} Kari, J.: Theory of cellular automata: a survey. Theoretical Computer Science \textbf{334} (2005) 3 -- 33. \url{https://doi.org/10.1016/j.tcs.2004.11.021}

\bibitem{Kurka} K\r{u}rka, P.: Topological and Symbolic Dynamics. Soci{\'e}t{\'e} math{\'e}matique de France, 2003.

\bibitem{LM95} Lind, D., Marcus, B.: An introduction to symbolic dynamics and coding. 2nd Ed. Cambridge University Press, 2021.


\end{thebibliography}
\end{document}